\documentclass[12pt]{article}
\usepackage[letterpaper, margin=1in]{geometry}
\usepackage{amsmath,amsthm,amssymb}
\usepackage{braket}
\usepackage{mathtools}
\usepackage[colorlinks,citecolor=blue,linkcolor=blue,urlcolor=blue,anchorcolor=blue]{hyperref}
\usepackage[table]{xcolor}
\usepackage{multirow}
\usepackage{makecell}
\usepackage{authblk}

\theoremstyle{plain}
\newtheorem{theorem}{Theorem}
\newtheorem{lemma}[theorem]{Lemma}
\newtheorem{corollary}[theorem]{Corollary}
\newtheorem{fact}[theorem]{Fact}
\newtheorem*{restatedcorollarytwo}{Corollary 2}
\newtheorem*{restatedcorollarythree}{Corollary 3}

\theoremstyle{definition}
\newtheorem{definition}[theorem]{Definition}

\theoremstyle{remark}

\title{ Fanout Complexity of Symmetric Boolean Functions in $\mathsf{QAC}^0$}

\author{Boyan Xu\thanks{Email: xuby26@mail2.sysu.edu.cn}~ and Lvzhou Li\thanks{Email: lilvzh@mail.sysu.edu.cn (corresponding author)}\\
\small{{\it Institute of Quantum Computing and Software,}}
\small{{\it School of Computer Science and Engineering,}}\\
\small {{\it  Sun Yat-sen University, Guangzhou 510006, China}}}

\begin{document}

\maketitle
\begin{abstract}
Whether $\mathsf{QAC}^0$ can compute $\mathtt{PARITY}_n$ remains open. Computing $\mathtt{PARITY}_n$ is equivalent to implementing $\mathtt{FANOUT}_n$ under $\mathsf{QAC}^0$ reductions. This raises a more general question: for an arbitrary symmetric Boolean function $f:\{0,1\}^n\to\{0,1\}$, what fanout size is necessary and sufficient for computing $f$ in $\mathsf{QAC}^0$? We show that the answer is exactly the transition radius $\rho(f)$: computing $f$ and implementing $\mathtt{FANOUT}_{\rho(f)}$ are equivalent under $\mathsf{QAC}^0$ reductions. In particular, if $\rho(f)\ge n^\delta$ for some constant $\delta>0$, then computing $f$ is $\mathsf{QAC}^0_{\mathrm{f}}$-complete. Combined with Paturi's theorem, our characterization implies that if $\mathtt{PARITY}_n \notin \mathsf{QAC}^0$, then any Boolean function in $\mathsf{QAC}^0$ of approximate degree $n^{1/2+\Omega(1)}$ must be nonsymmetric.
\end{abstract}

\section{Introduction}

$\mathsf{QAC}^0$ is the class of constant-depth, polynomial-size quantum circuits built from arbitrary single-qubit gates and multi-controlled Toffoli gates. It was introduced by Moore as a quantum analogue of $\mathsf{AC}^0$ \cite{Moore1999FanoutParityCounting,GreenHomerMoorePollett2002CountingFanout}. A longstanding question is whether $\mathsf{QAC}^0$ can implement the $\mathtt{FANOUT}_n$ operation, which copies classical information from one control qubit to $n$ target qubits. In classical circuits, fanout is freely available: the output of a gate can be fed into arbitrarily many other gates. In quantum circuits, by contrast, $\mathtt{FANOUT}$ is a genuine circuit resource. We denote the class obtained by adjoining $\mathtt{FANOUT}_n$ to $\mathsf{QAC}^0$ by $\mathsf{QAC}^0_{\mathrm{f}} \coloneqq \mathsf{QAC}^0[\mathtt{FANOUT}_n]$. Equivalently, the question is whether $\mathsf{QAC}^0=\mathsf{QAC}^0_{\mathrm{f}}$.

The $\mathtt{FANOUT}$ and $\mathtt{PARITY}$ gates are equivalent up to Hadamard conjugation \cite{Moore1999FanoutParityCounting}. Thus, adjoining $\mathtt{PARITY}$ to $\mathsf{QAC}^0$ gives exactly the class $\mathsf{QAC}^0_{\mathrm{f}}$. Classically, adjoining parity, modular, and threshold gates to $\mathsf{AC}^0$ respectively leads to the hierarchy $\mathsf{AC}^0\subsetneq\mathsf{AC}^0[2]\subsetneq\mathsf{ACC}^0\subseteq\mathsf{TC}^0$ \cite{Ajtai1983Sigma11,FurstSaxeSipser1984Parity,Hastad1989AlmostOptimal,Razborov1987LowerBounds,Smolensky1987Algebraic}. In the quantum setting, by contrast, these distinctions collapse: adjoining modular or threshold gates to $\mathsf{QAC}^0$ yields the same class as adjoining parity or fanout; in particular, $\mathsf{QAC}^0_{\mathrm{f}} = \mathsf{QTC}^0$ \cite{GrierMorris2025QuantumThreshold,GrierMorrisWu2026QACContainsTC}. The class $\mathsf{QTC}^0$ contains all classical $\mathsf{TC}^0$ computations, including threshold computation, counting, sorting, and arithmetic \cite{HoyerSpalek2005QuantumFanout, TakahashiTani2016Collapse}. Hence the fanout question is all-or-nothing: a positive answer would imply $\mathsf{QAC}^0=\mathsf{QTC}^0$, whereas a negative answer would give the natural quantum counterpart of the classical parity lower bound. Despite substantial effort, the question has remained open for nearly three decades \cite{FangFennerGreenHomerZhang2006Quantum,Bera2011LowerBound,Rosenthal2021ApproximatingParity,NadimpalliEtAl2024PauliSpectrum,FennerGrierPadeThierauf2025Depth2,AnshuEtAl2025BarelySuperlinear, FoxmanParhamVasconcelosYuen2026RandomUnitaries, JoshiTalVasconcelosWright2026Improved,GrettaGuptaJoshi2026FourierConcentrated,DongOuYao2026GeometricallyLocalQAC0}.

While the longstanding question concerns $\mathtt{FANOUT}$ of size $n$, recent work on Dicke-state preparation suggests a finer-grained perspective based on the size of the available fanout gate. The $n$-qubit weight-$k$ Dicke state $\ket{D^n_k}$ is the uniform superposition over $n$-bit strings of Hamming weight $k$. Gretta, Gupta, and Joshi showed that $\mathtt{FANOUT}_k$ is necessary for preparing $\ket{D^n_k}$ in $\mathsf{QAC}^0$ for $k \le n/2$ \cite{GrettaGuptaJoshi2026FourierConcentrated}. In a companion work, they established the matching upper bound: $\ket{D^n_k}$ can be prepared by $\mathsf{QAC}^0[\mathtt{FANOUT}_k]$ circuits \cite{GrettaGuptaJoshi2026SuperConstantDicke}, building on the work of Joshi and Vasconcelos \cite{JoshiVasconcelos2026DickeStates}. Together, these results establish an exact correspondence between preparing the weight-$k$ Dicke state and implementing $\mathtt{FANOUT}_k$ for $k\le n/2$. Our main result establishes an analogous correspondence for the entire class of symmetric Boolean functions. For every symmetric Boolean function $f$, we characterize the corresponding fanout size in terms of a single structural parameter: its transition radius $\rho(f)$. More precisely, we prove that computing $f$ and implementing $\mathtt{FANOUT}_{\rho(f)}$ are equivalent under $\mathsf{QAC}^0$ reductions.

\subsection{Our results}\label{sec:intro-results}

A Boolean function $f:\{0,1\}^n\to\{0,1\}$ is \emph{symmetric} if its value depends only on the Hamming weight of the input. For $i\in\{0,\ldots,n\}$, let $f_i$ denote the value of $f$ on inputs of Hamming weight $i$. A \emph{transition} of $f$ occurs between two consecutive Hamming weights on which its value changes. We write
$$
\Delta_f
\coloneqq
\{i\in\{0,\ldots,n-1\}:f_i\ne f_{i+1}\}
$$
for the set of transitions and define the \emph{transition radius} of $f$ by
$$
\rho(f)
\coloneqq
\max_{i\in \Delta_f}\min\{i+1,n-i\}.
$$
For a constant function, we set $\rho(f)=0$ and regard $\mathtt{FANOUT}_0$ as the identity operation. Thus, $\rho(f)$ records how far the innermost transition lies from the nearer endpoint of the Hamming-weight interval. For example, $\rho(\mathtt{AND}_n)=\rho(\mathtt{OR}_n)=1$, whereas both parity and majority have transition radius $\lceil n/2\rceil$.

We show that the transition radius provides a precise correspondence between computing a symmetric function and implementing fanout of the corresponding size.

\begin{theorem}\label{thm:intro-main}
Let $f:\{0,1\}^n\to\{0,1\}$ be symmetric. Then:
\begin{enumerate}
    \item For any constant $c > 0$, if $f$ can be computed by a $\mathsf{QAC}^0$ circuit with probability at least $1/2 + 1/\log^c n$ on every input, then $\mathtt{FANOUT}_{\rho(f)}$ can be implemented by a $\mathsf{QAC}^0$ circuit.
    \item Conversely, if $\mathtt{FANOUT}_{\rho(f)}$ can be implemented by a $\mathsf{QAC}^0$ circuit, then $f$ can be computed by a $\mathsf{QAC}^0$ circuit exactly.
\end{enumerate}
In particular, we have:
$$
f \in \mathsf{QAC}^0 \Longleftrightarrow \mathtt{FANOUT}_{\rho(f)} \in \mathsf{QAC}^0.
$$
\end{theorem}

Here, $f\in\mathsf{QAC}^0$ means that $f$ can be computed exactly by a $\mathsf{QAC}^0$ circuit, while $\mathtt{FANOUT}_{\rho(f)}\in\mathsf{QAC}^0$ means that this unitary can be implemented exactly and cleanly by a $\mathsf{QAC}^0$ circuit.

When $\rho(f)=O(\log^d n)$ for some constant $d>0$, the classical characterization of symmetric $\mathsf{AC}^0$ functions shows that this is precisely the regime in which a symmetric function lies in $\mathsf{AC}^0$ \cite{Moran1987GeneralizedLowerBounds,BrustmannWegener1987SymmetricFunctions}. The corresponding $\mathsf{QAC}^0$ upper bound was already established by Grier, Morris, and Wu, who showed that polylogarithmic-size fanout is available in $\mathsf{QAC}^0$ and used it to prove that symmetric $\mathsf{AC}^0$ functions are contained in $\mathsf{QAC}^0$ \cite{GrierMorrisWu2026QACContainsTC}.

When $\rho(f)\ge n^\delta$ for some constant $\delta>0$, the correspondence yields a $\mathsf{QAC}^0_{\mathrm{f}}$-completeness result. By iterating $\mathtt{FANOUT}_{\rho(f)}$ in a tree of constant depth $O(1/\delta)$, one can implement $\mathtt{FANOUT}_n$. Hence, by Theorem~\ref{thm:intro-main}, any $\mathsf{QAC}^0$ circuit that computes $f$ with probability at least $1/2+1/\log^c n$ on every input can be used to implement $\mathtt{FANOUT}_n$, for any fixed constant $c>0$. Conversely, $\mathtt{FANOUT}_n$ suffices to compute $f$ exactly. Following Gretta, Gupta, and Joshi \cite{GrettaGuptaJoshi2026FourierConcentrated}, a task is said to be $\mathsf{QAC}^0_{\mathrm{f}}$-complete if it is equivalent to  $\mathtt{FANOUT}_n$ under $\mathsf{QAC}^0$ reductions. This gives the following corollary.

\begin{corollary}\label{cor:intro-qacf-complete}
Let $f:\{0,1\}^n\to\{0,1\}$ be symmetric, and fix any constant $c>0$. If $\rho(f)\ge n^\delta$ for some constant $\delta>0$, then computing $f$ with probability at least $1/2+1/\log^c n$ on every input is $\mathsf{QAC}^0_{\mathrm{f}}$-complete.
\end{corollary}

Table~\ref{tab:prior} places our result in the context of known $\mathsf{QAC}^0_{\mathrm{f}}$-completeness results for symmetric functions. Several previously separate results for $\mathtt{PARITY}$, $\mathtt{MAJORITY}$, $\mathtt{THRESHOLD}$, $\mathtt{EXACT}$, and $\mathtt{MOD}_q$ are captured by a single criterion based on transition radius: symmetric functions with polynomial transition radius are complete when computed with probability at least $1/2+1/\mathsf{polylog}(n)$ on every input. The known average-case completeness results are not implied by our theorem.

\begin{table}[!t]
\centering
\footnotesize
    \renewcommand{\arraystretch}{1.5}
    \setlength{\tabcolsep}{1.2pt}
    \begin{tabular}{|c|c|c|c|c|}
\hline
\rowcolor{gray!10} Function $f$ & $\rho(f)$ &  Computation & Work & \makebox[\widthof{\textbf{Strengthened}}][c]{
    \makecell[c]{Relation to\\[-1pt]this work}}\\
\hline 
\multirow{3}{*}{$\mathtt{PARITY}_n$}
& \multirow{3}{*}{$\lceil n/2 \rceil$}
& Exact & \cite{Moore1999FanoutParityCounting} & Recovered \\
&
& Average-case correlation $\ge 1/\mathsf{polylog}(n)$ & \cite{JoshiTalVasconcelosWright2026Improved} & Not covered\\
&
& Average-case correlation $\ge 1/\mathsf{poly}(n)$ & \cite{GrettaGuptaJoshi2026FourierConcentrated} & Not covered\\
\hline
\multirow{3}{*}{$\mathtt{MAJORITY}_n$}
& \multirow{3}{*}{$\lceil n/2 \rceil$}
& Exact & \cite{GrierMorris2025QuantumThreshold,GrierMorrisWu2026QACContainsTC} & Recovered \\
&
& Worst-case probability $\ge 1/2 + 1/\mathsf{polylog}(n)$ & \cite{GrettaGuptaJoshi2026FourierConcentrated} & \textbf{Generalized}\\
&
& Average-case correlation $\ge 1-1/\mathsf{poly}(n)$ & \cite{GrettaGuptaJoshi2026FourierConcentrated} & Not covered\\
\hline
\multirow{2}{*}{\makecell[c]{$\mathtt{THRESHOLD}^n_k$\\[-1pt]
$(n^{\Omega(1)}\le k\le \lceil n/2 \rceil)$}}
&
\multirow{2}{*}{$k$}
& Exact &
\cite{GrierMorris2025QuantumThreshold,GrierMorrisWu2026QACContainsTC} & \textbf{Strengthened}\\
&
& Unitary approximation error $\le 1/\mathsf{poly}(n)$ &
\cite{GrierMorris2025QuantumThreshold,GrierMorrisWu2026QACContainsTC} & \textbf{Strengthened}\\
\hline
\multirow{2}{*}{$\mathtt{EXACT}^n_{\lceil n/2 \rceil}$}
&
\multirow{2}{*}{$\lceil n/2 \rceil$}
& Exact &
\cite{GrierMorris2025QuantumThreshold,GrierMorrisWu2026QACContainsTC} & \textbf{Strengthened}\\
&
& Unitary approximation error $\le 1/\mathsf{poly}(n)$ &
\cite{GrierMorris2025QuantumThreshold,GrierMorrisWu2026QACContainsTC} & \textbf{Strengthened}\\
\hline
$\mathtt{MOD}^n_{q}$, fixed $q > 2$
& $\Theta(n)$
& Exact  & \cite{GreenHomerMoorePollett2002CountingFanout} & \textbf{Strengthened} \\
\hline
\makecell[c]{Symmetric $f$ with \\[-1pt]
$n^{\Omega(1)} \le \rho(f) \le \lceil n/2 \rceil$}
& $n^{\Omega(1)}$
& Worst-case probability $\ge 1/2 + 1/\mathsf{polylog}(n)$ & \textbf{This work} & --\\
\hline
\end{tabular}
\caption{Known $\mathsf{QAC}^0_{\mathrm{f}}$-completeness results for symmetric Boolean functions. In the last column, ``Recovered'' indicates that the corresponding prior result follows as a special case of our theorem; ``Generalized'' indicates that our theorem extends the corresponding statement from a specific function to symmetric functions of polynomial transition radius; and ``Strengthened'' indicates that our theorem yields the same completeness conclusion under a weaker computational assumption.}
\label{tab:prior} 
\end{table}

Our characterization also yields consequences for approximate degree, a standard complexity measure of Boolean functions. Recall that the approximate degree $\widetilde{\deg}(f)$ is the minimum degree of a real polynomial that approximates $f$ pointwise to error at most $1/3$. Bun and Thaler showed that, for every constant $\delta>0$, there exists $g\in\mathsf{AC}^0$ with $\widetilde{\deg}(g)=\Omega(n^{1-\delta})$ \cite{BunThaler2017NearlyOptimal}. It remains open whether $\mathsf{QAC}^0$ contains Boolean functions of large approximate degree. For nonconstant symmetric functions, Paturi's characterization \cite{Paturi1992ApproximateSymmetric}, together with the definition of transition radius, gives
$$
\widetilde{\deg}(f)
=
\Theta\!\left(\sqrt{n\rho(f)}\right).
$$
Thus, approximate degree $n^{1/2 + \Omega(1)}$ forces polynomially large transition radius, and Corollary~\ref{cor:intro-qacf-complete} gives the following result.

\begin{corollary}\label{cor:intro-approx-degree}
Let $f:\{0,1\}^n\to\{0,1\}$ be symmetric, and suppose that, for some constant $\eta > 0$, $\widetilde{\deg}(f)=\Omega\!\left(n^{1/2+\eta}\right)$. Then, for every constant $c > 0$, computing $f$ with probability at least $1/2+1/\log^c n$ on every input is $\mathsf{QAC}^0_{\mathrm{f}}$-complete.
\end{corollary}

Indeed, Paturi's characterization implies $\rho(f)=\Omega(n^{2\eta})$, so the result follows directly from Corollary~\ref{cor:intro-qacf-complete}.  Consequently, if $\mathtt{PARITY}_n \notin \mathsf{QAC}^0$, then any function in $\mathsf{QAC}^0$ of approximate degree $n^{1/2+\Omega(1)}$ must be nonsymmetric.

\subsection{Technical overview}

We give an overview of the proof of Theorem~\ref{thm:intro-main}.

\paragraph{From fanout to the symmetric function.}
This direction is a parameterized version of the exact-weight construction of Grier, Morris, and Wu~\cite{GrierMorrisWu2026QACContainsTC}. Let $k=\rho(f)$. Since all transitions of $f$ occur among the first or last $k$ Hamming-weight layers, up to a possible negation of the output, $f$ is an $\mathtt{OR}$ of at most $2k$ predicates $\mathtt{EXACT}_j^n$, with $j<k$ or $j>n-k$. Using $\mathtt{FANOUT}_k$, the construction of Grier, Morris, and Wu evaluates these exact-weight predicates in parallel. A final $\mathtt{OR}$ gate, followed by an output negation if needed, computes $f$ exactly.

\paragraph{From the symmetric function to fanout.}

The new ingredient is to convert a transition at radius $k$ into the Fourier nonconcentration needed to produce fanout. Suppose that a circuit $C$ computes $f$ with probability at least $1/2+\varepsilon$ on every input, where $\varepsilon \ge 1/\log^c n$. We may assume that $k \ge \log^A n$ for some constant $A > 12c$, since the polylogarithmic case follows from Fact~\ref{fact:polylog-fanout}.

Choose a transition attaining $\rho(f) = k$. After negating all input bits if necessary, we may assume that it occurs between Hamming weights $k-1$ and $k$. Keeping $m=2k-1$ input bits and fixing the rest to zero gives a restricted symmetric function $f'$ with this transition at its center. Let $C'$ be the corresponding restriction of $C$, and define
$$
g(x) \coloneqq \Pr[C'(x) = 0] - \Pr[C'(x) = 1].
$$
Since $C$ computes $f$ with probability at least $1/2+\varepsilon$ on every input, we have
$$
(-1)^{f'(x)}g(x) \ge 2\varepsilon \qquad \text{for every } x \in \{0,1\}^m.
$$
Hence, for every adjacent pair $x,y$ with $|x| = k-1$ and $|y| = k$,
$$
|g(x) - g(y)| \ge 4\varepsilon.
$$
There are $\Theta(2^m \sqrt{k})$ such pairs, and summing over them gives
$$
\sum_{S \subseteq [m]}{|S| \widehat{g}(S)^2} = \Omega(\varepsilon^2 \sqrt{k}).
$$
Consequently, for a suitable $\ell = \Theta(\varepsilon^2 \sqrt{k})$,
$$
\sum_{|S| \ge \ell}{\widehat{g}(S)^2} = \Omega\left(\frac{\varepsilon^2}{\sqrt{k}}\right).
$$
Since $m=2k-1$, $k \ge \log^A n$ and $\varepsilon \ge 1/\log^c n$, we have $\ell \ge m^{1/3}$, while the Fourier mass above is inverse polynomial. The Fourier-tail reduction of Gretta, Gupta, and Joshi~\cite{GrettaGuptaJoshi2026FourierConcentrated} therefore yields $\mathtt{FANOUT}_\ell$. Finally, $\ell^3\ge m\ge k$, so a three-level fanout tree implements $\mathtt{FANOUT}_k$.

\section{Preliminaries}\label{sec:preliminaries}

For an integer $n\ge1$, write $[n]\coloneqq\{1,\ldots,n\}$. For a bit string $x\in\{0,1\}^n$, let $|x|$ denote its Hamming weight.

\subsection{Analysis of Boolean functions}

We briefly recall the Fourier-analytic notation used below; see O'Donnell~\cite{ODonnell2014Analysis} for further background.

For $f,g:\{0,1\}^n\to\mathbb{R}$, define their correlation by
$$
\operatorname{corr}(f,g)
\coloneqq
\mathop{\mathbb{E}}_{x\sim\{0,1\}^n}[f(x)g(x)].
$$
Here $x\sim\{0,1\}^n$ means that $x$ is sampled uniformly at random from $\{0,1\}^n$. For each $S\subseteq[n]$, define the \emph{Fourier character}
$$
\chi_S(x)\coloneqq(-1)^{\sum_{j\in S}x_j}.
$$
Every function $f:\{0,1\}^n\to\mathbb{R}$ has the following unique \emph{Fourier expansion}
$$
f(x)=\sum_{S\subseteq[n]}\widehat{f}(S)\chi_S(x),
$$
where 
$$
\widehat{f}(S)
\coloneqq \operatorname{corr}(f,\chi_S)
=\mathop{\mathbb{E}}_{x\sim\{0,1\}^n}\bigl[f(x)\chi_S(x)\bigr]
$$
are called the \emph{Fourier coefficients} of $f$.

\begin{definition}[Degree]
    The \emph{degree} of a function $f:\{0,1\}^n\to\mathbb{R}$ is
    $$
    \deg(f)
    \coloneqq
    \max\bigl\{|S|:\widehat{f}(S)\ne0\bigr\}.
    $$
    We adopt the convention $\deg(0)=0$.
\end{definition}

\begin{definition}[Approximate degree]
    Let $f:\{0,1\}^n\to\mathbb{R}$ and $\varepsilon\ge0$. The \emph{approximate degree of $f$ with error $\varepsilon$} is
    $$
    \widetilde{\deg}_{\varepsilon}(f)
    \coloneqq
    \min\bigl\{\deg(g):g:\{0,1\}^n\to\mathbb{R},
    \ |g(x)-f(x)|\le\varepsilon\text{ for every }x\in\{0,1\}^n\bigr\}.
    $$
    We abbreviate $\widetilde{\deg}_{1/3}(f)$ as $\widetilde{\deg}(f)$.
\end{definition}

For $0\le d, \ell \le n$, define the \emph{level-$d$ Fourier weight} and the \emph{Fourier weight above level $\ell$} by
$$
\mathbf{W}^d[f]\coloneqq\sum_{|S|=d}\widehat{f}(S)^2,
\qquad
\mathbf{W}^{\ge\ell}[f]\coloneqq\sum_{d=\ell}^n\mathbf{W}^d[f].
$$
Parseval's identity gives
$$
\sum_{d=0}^n\mathbf{W}^d[f]=\mathbb{E}_x[f(x)^2].
$$

The \emph{influence} of coordinate $j \in [n]$ on $f:\{0,1\}^n\to\mathbb{R}$ is defined by 
$$
\mathbf{Inf}_j[f] \coloneqq \frac{1}{4}\mathop{\mathbb{E}}\limits_{x \sim \{0,1\}^n}[(f(x)-f(x \oplus e_j))^2],
$$
where $e_j$ is the $j$th standard basis vector. The \emph{total influence} of $f$ is
$$
\mathbf{I}[f]\coloneqq\sum_{j=1}^n\mathbf{Inf}_j[f].
$$

\begin{fact}[\cite{ODonnell2014Analysis}]\label{fact:influence-fourier}
    For every function $f : \{0,1\}^n \to \mathbb{R}$, we have
    $$
    \mathbf{I}[f]=\sum_{d=1}^n d\,\mathbf{W}^{d}[f].
    $$
\end{fact}

\subsection{Symmetric Boolean functions}

A Boolean function $f:\{0,1\}^n\to\{0,1\}$ is \emph{symmetric} if it is invariant under permutations of its input bits, or equivalently, if $f(x)$ depends only on $|x|$. For a symmetric function $f$, write $f_i$ for its common value on inputs of Hamming weight $i$. The \emph{transition set} of $f$ is
$$
\Delta_f
\coloneqq
\{i\in\{0,\ldots,n-1\}:f_i\ne f_{i+1}\}.
$$
If $\Delta_f\ne\varnothing$, define the \emph{transition radius} of $f$ by
$$
\rho(f)
\coloneqq
\max_{i\in \Delta_f}\min\{i+1,n-i\};
$$
for a constant function ($\Delta_f = \varnothing$), set $\rho(f)=0$. A Boolean function family $(f^{(n)})_{n\ge1}$ is symmetric if every $f^{(n)}$ is symmetric. We define symmetry, the transition set, and the transition radius similarly for $\{\pm 1\}$-valued functions.

For a nonconstant symmetric function $f$, Paturi's jump parameter is defined by
$$
\Gamma(f)
\coloneqq
\min\bigl\{|2i-n+1|:i\in\Delta_f\bigr\}.
$$

\begin{fact}[Paturi~\cite{Paturi1992ApproximateSymmetric}]\label{fact:paturi}
    For every nonconstant symmetric Boolean function $f:\{0,1\}^n\to\{0,1\}$,
    $$
    \widetilde{\deg}(f)
    =
    \Theta\!\left(\sqrt{n\bigl(n-\Gamma(f)\bigr)}\right).
    $$
\end{fact}

The two transition parameters are related by
$$
\Gamma(f)=n+1-2\rho(f).
$$
Indeed, for every $0\le i\le n-1$,
$$
|2i-n+1|=n+1-2\min\{i+1,n-i\},
$$
so minimizing over $i\in\Delta_f$ gives the identity above.

For $0\le j\le n$, define
$$
\mathtt{EXACT}^n_j(x)\coloneqq\mathbf{1}[|x|=j],
\qquad
\mathtt{THRESHOLD}^n_j(x)\coloneqq\mathbf{1}[|x|\ge j].
$$
For an integer $q\ge2$, define
$$
\mathtt{MOD}^n_q(x)
\coloneqq
\mathbf{1}[|x|\equiv 0\pmod q].
$$
We also define
$$
\mathtt{PARITY}_n(x)
\coloneqq
\bigoplus_{i=1}^n x_i
=\mathbf{1}[|x|\text{ is odd}],
\qquad
\mathtt{MAJORITY}_n(x)
\coloneqq
\mathbf{1}[|x|>n/2].
$$
Here $\mathbf{1}[\mathcal{E}]$ denotes the indicator of the condition $\mathcal{E}$. In particular,
$$
\mathtt{OR}_n=\mathtt{THRESHOLD}^n_1,
\qquad
\mathtt{AND}_n=\mathtt{EXACT}^n_n=\mathtt{THRESHOLD}^n_n.
$$
We omit the input length $n$ when it is clear from context.

\subsection{Circuit model}

For every Boolean function $f:\{0,1\}^n\to\{0,1\}$, the \emph{reversible gate associated with $f$} is the $(n+1)$-qubit unitary $U_f$ defined by
$$
U_f\ket{x,b}
=
\ket{x,b\oplus f(x)}.
$$

For an integer $r\ge1$, the $r$-controlled Toffoli gate is the reversible gate $U_{\mathtt{AND}_r}$ associated with the $r$-bit $\mathtt{AND}$ function; explicitly,
$$
\mathtt{TOFFOLI}_r
\ket{x_1,\ldots,x_r,b}
=
\ket{x_1,\ldots,x_r,b\oplus(x_1\wedge\cdots\wedge x_r)}.
$$
The case $r=1$ is a $\mathtt{CNOT}$ gate. The $\mathtt{FANOUT}$ gate on one control and $r$ targets is the $(r+1)$-qubit unitary
$$
\mathtt{FANOUT}_r
\ket{b,x_1,\ldots,x_r}
=
\ket{b,x_1\oplus b,\ldots,x_r\oplus b}.
$$
We regard $\mathtt{FANOUT}_0$ as the identity operation.

A \emph{$\mathsf{QAC}$ circuit} is a quantum circuit composed of arbitrary single-qubit unitaries and multi-controlled Toffoli gates. We adopt the depth convention used in prior work. Namely, a $\mathsf{QAC}$ circuit $C$ has depth $d$ if it admits a decomposition
$$
C=L_0M_1L_1M_2L_2\cdots M_dL_d,
$$
where each $L_i$ is a layer of single-qubit gates and each $M_i$ is a layer of multi-controlled Toffoli gates acting on pairwise disjoint sets of qubits. Thus, the single-qubit layers do not contribute to the depth. 

Let $C$ be a $\mathsf{QAC}$ circuit acting on $n$ input qubits and $a$ ancillae, with the last ancilla designated as the target qubit $t$. For each $x\in\{0,1\}^n$, let $C(x)\in\{0,1\}$ denote the outcome of measuring $t$ in the computational basis after applying $C$ to $\ket{x}\ket{0^{a-1}}\ket{0}_t$. For $b\in\{0,1\}$, the distribution of $C(x)$ is given by
$$
p_b(x) \coloneqq \Pr[C(x)=b]
=
\left\|
\bigl(\mathbb{I} \otimes\bra{b}_t\bigr)
C\ket{x}\ket{0^{a-1}}\ket{0}_t
\right\|_2^2,
$$
where $\mathbb{I}$ acts on all qubits other than $t$. The associated \emph{real-valued Boolean function computed by $C$} is
$$
g_C(x) \coloneqq p_0(x)-p_1(x).
$$

Let $f:\{0,1\}^n \to \{0,1\}$ be a Boolean function. We say that the circuit $C$ \emph{computes $f$ exactly} if 
$$
\Pr[C(x) = f(x)] = 1
\qquad
\text{for every }x\in\{0,1\}^n.
$$
For $p\in[0,1]$, we say that $C$ \emph{computes $f$ with
worst-case probability at least $p$} if
$$
\Pr[C(x)=f(x)]\ge p
\qquad
\text{for every }x\in\{0,1\}^n.
$$
In particular, for $\varepsilon\in[0,1/2]$, $C$ computes $f$
with worst-case probability at least $1/2+\varepsilon$ if and only if
$$
g_C(x)(-1)^{f(x)}\ge2\varepsilon
\qquad
\text{for every }x\in\{0,1\}^n.
$$

\begin{definition}[$\mathsf{QAC}^0$ circuit family]\label{def:qac0-family}
    For each $n \ge 1$, let $C_n$ be a $\mathsf{QAC}$ circuit of depth $d(n)$ acting on $n$ input qubits and $a(n)$ ancillae. The circuit family $\mathcal{C}=(C_n)_{n\ge1}$ is called a $\mathsf{QAC}^0$ circuit family if there exist constants $d_0,\alpha>0$ and $n_0\ge1$ such that, for all $n\ge n_0$, it holds that $a(n) \le n^\alpha$ and $d(n) \le d_0$.
\end{definition}

Given a Boolean function family $\mathcal{F}=(f^{(n)})_{n\ge1}$, where each $f^{(n)}:\{0,1\}^n\to\{0,1\}$, and a circuit family $\mathcal{C}=(C_n)_{n\ge1}$, we say that $\mathcal{C}$ \emph{computes $\mathcal{F}$ exactly} if $C_n$ computes $f^{(n)}$ exactly for every $n\ge1$. Computing a function family with worst-case probability at least $p(n)$ is defined analogously. We write $\mathcal{F}\in\mathsf{QAC}^0$ if some $\mathsf{QAC}^0$ circuit family computes $\mathcal{F}$ exactly.

A $\mathsf{QAC}$ circuit $C$ \emph{implements} an $m$-qubit unitary $U$ using $a$ ancillae if, for every $m$-qubit state $\ket{\psi}$,
$$
C\ket{\psi}\ket{0^a}=\bigl(U\ket{\psi}\bigr)\ket{0^a}.
$$

Let $k=k(n)$ and $0 \le k(n) \le n$. We write $\mathtt{FANOUT}_k\in\mathsf{QAC}^0$ if there is a constant-depth family $(C_n)_{n\ge1}$ of $\mathsf{QAC}$ circuits using $\mathsf{poly}(n)$ ancillae such that $C_n$ implements the $(n+1)$-qubit unitary $\mathtt{FANOUT}_{k(n)}\otimes\mathbb{I}_{n-k(n)}$ for every $n\ge1$.

\begin{definition}[{$\mathsf{QAC}^0[\mathtt{FANOUT}_k]$ circuit family}]\label{def:qac0-with-fanout}
    Let $k=k(n)$ satisfy $0\le k(n)\le n$. For each $n\ge1$, let $C_n$ be a quantum circuit of depth $d(n)$ acting on $n$ input qubits and $a(n)$ ancillae, composed of arbitrary single-qubit unitaries, multi-controlled Toffoli gates, and $\mathtt{FANOUT}_{k(n)}$ gates, where depth is measured by counting layers of multi-qubit gates as above. The circuit family $\mathcal{C}=(C_n)_{n\ge1}$ is called a $\mathsf{QAC}^0[\mathtt{FANOUT}_k]$ circuit family if there exist constants $d_0,\alpha>0$ and $n_0\ge1$ such that, for all $n\ge n_0$, it holds that $a(n)\le n^\alpha$ and $d(n)\le d_0$.
\end{definition}

The notation $\mathcal{F}\in\mathsf{QAC}^0[\mathtt{FANOUT}_k]$ is defined analogously to $\mathcal{F}\in\mathsf{QAC}^0$. We also write $\mathsf{QAC}^0_{\mathrm{f}}\coloneqq\mathsf{QAC}^0[\mathtt{FANOUT}_n]$.

\begin{fact}[\cite{HoyerSpalek2005QuantumFanout,TakahashiTani2016Collapse}]\label{fact:symmetric-in-qacf}
    Let $\mathcal{F} \coloneqq (f^{(n)})_{n \ge1}$, where each $f^{(n)}:\{0,1\}^n \to \{0,1\}$ is symmetric. Then $\mathcal{F} \in \mathsf{QAC}^0_{\mathrm{f}}$.
\end{fact}

\begin{fact}[\cite{GrierMorrisWu2026QACContainsTC}]\label{fact:polylog-fanout}
    Let $k=k(n)$ satisfy $k=O(\log^d n)$ for some constant $d>0$. Then
    $\mathtt{FANOUT}_k\in\mathsf{QAC}^0$.
\end{fact}

\begin{fact}[\cite{GrierMorrisWu2026QACContainsTC}]\label{fact:exact-from-fanout}
    Let $k=k(n)$ satisfy $0\le k(n)\le n/2$. Then
    $\mathtt{EXACT}^n_k\in
    \mathsf{QAC}^0[\mathtt{FANOUT}_k]$.
\end{fact}

The proof of Fact~\ref{fact:exact-from-fanout} is provided in Appendix~\ref{app:exact-from-fanout} for completeness.

\section{The reductions}

We prove Theorem~\ref{thm:intro-main} by giving explicit $\mathsf{QAC}^0$ reductions in both directions. We first show that $\mathtt{FANOUT}_k$ suffices to compute every symmetric function with transition radius $k$. We then show conversely that any circuit computing such a function with worst-case probability at least $1/2+1/\mathsf{polylog}(n)$ can be transformed into a circuit implementing $\mathtt{FANOUT}_k$. Finally, we derive the two corollaries stated in the introduction.

\subsection{From fanout to symmetric functions}

We begin with the reduction from fanout to symmetric-function computation. The proof is a parameterized version of the exact-weight decomposition used by Grier, Morris, and Wu in their proof of Corollary~16~\cite{GrierMorrisWu2026QACContainsTC}.

\begin{theorem}\label{thm:upper-bound}
    Let $f : \{0,1\}^n \to \{0,1\}$ be a symmetric function with $\rho(f) = k$. Then
    $$
    f \in \mathsf{QAC}^0[\mathtt{FANOUT}_k].
    $$
\end{theorem}

\begin{proof}
    The case $k=0$ is immediate, so assume that $k\ge1$. Since $\rho(f)=k$, the values $f_k,f_{k+1},\ldots,f_{n-k}$ are all equal. By negating the output if necessary, we may assume these values are all zero. Then $f$ can be expressed as
    $$
    f(x)=\bigvee_{j\in S}\mathtt{EXACT}_{j}(x)
    $$
    for some
    $$
    S \subseteq \{0,\ldots,k-1\} \cup \{n-k+1,\ldots,n\}.
    $$
    In particular, $|S| \le 2k$. By Fact~\ref{fact:exact-from-fanout}, we know that $\mathtt{EXACT}_{j} \in \mathsf{QAC}^0[\mathtt{FANOUT}_k]$ for all $j < k$, since $\mathtt{FANOUT}_j \in \mathsf{QAC}^0[\mathtt{FANOUT}_k]$. For $j > n-k$, we use 
    $$
    \mathtt{EXACT}_{j}(x)=\mathtt{EXACT}_{n-j}(\neg x),
    $$
    which reduces to an exact-weight function with $n-j<k$. Hence it suffices to make $2k$ copies of the input string $x$ using two parallel layers of $\mathtt{FANOUT}_k$, then compute the corresponding $\mathtt{EXACT}_{j}(x)$ in parallel, and finally combine their outputs with a single $\mathtt{OR}$ gate to compute $f$. The resulting circuit has constant depth and polynomial size and computes $f$ exactly.
\end{proof}

\subsection{From symmetric functions to fanout}

We next prove the reverse reduction.

\begin{lemma}\label{lem:real-valued-central-fourier-tail}
    Let $k\ge1$ be an integer, set $m=2k-1$, and let $\gamma\in(0,1]$. Suppose that $h : \{0,1\}^m \to \{\pm1\}$ is symmetric with $\rho(h) = k$, and that $g:\{0,1\}^m \to [-1,1]$ satisfies
    $$
    g(x)h(x) \ge \gamma
    \qquad
    \text{for every }x\in\{0,1\}^m.
    $$
    Then, for $\ell = \left\lfloor \gamma^2\sqrt{m/8} \right\rfloor + 1$,
    $$
    \mathbf{W}^{\ge \ell}[g] \ge \frac{\gamma^2}{\sqrt{8m}}.
    $$
\end{lemma}

\begin{proof}
    Since $m=2k-1$, the assumption $\rho(h)=k$ implies that
    $k-1\in\Delta_h$. Fix $j\in[m]$. For every
    $x\in\{0,1\}^m$ with $|x|=k-1$ and $x_j=0$, set
    $y=x\oplus e_j$. Then $|y|=k$ and $h(y)=-h(x)$. Hence
    $$
    |g(x)-g(y)| = |h(x)\bigl(g(x)-g(y)\bigr)| = |g(x)h(x)+g(y)h(y)| \ge 2\gamma.
    $$

    By the definition of influence,
    $$
    \begin{aligned}
    \mathbf{Inf}_j[g]
    &=
    \frac{1}{2^{m+2}}
    \sum_{z\in\{0,1\}^m}
    \bigl(g(z)-g(z\oplus e_j)\bigr)^2\\
    &\ge
    \frac{1}{2^{m+2}}
    \sum_{\substack{x\in\{0,1\}^m\\|x|=k-1,\ x_j=0}}
    \left[
    \bigl(g(x)-g(x\oplus e_j)\bigr)^2
    +
    \bigl(g(x\oplus e_j)-g(x)\bigr)^2
    \right]\\
    &=
    \frac{1}{2^{m+1}}
    \sum_{\substack{x\in\{0,1\}^m\\|x|=k-1,\ x_j=0}}
    \bigl(g(x)-g(x\oplus e_j)\bigr)^2\\
    &\ge
    \gamma^2\frac{\binom{m-1}{k-1}}{2^{m-1}}.
    \end{aligned}
    $$
    Summing over $j\in[m]$ yields
    $$
    \mathbf{I}[g]
    \ge
    \gamma^2
    \frac{m\binom{m-1}{k-1}}{2^{m-1}}.
    $$

    If $m=1$, the right-hand side is $\gamma^2$, which is at least
    $\gamma^2\sqrt{m/2}$. Suppose now that $m\ge3$, and write
    $r=k-1=(m-1)/2$. Stirling's formula gives
    $$
    \frac{\binom{2r}{r}}{4^r}
    \ge
    \frac{1}{2\sqrt r}.
    $$
    It follows that
    $$
    \frac{m\binom{m-1}{k-1}}{2^{m-1}}
    \ge
    \frac{m}{\sqrt{2(m-1)}}
    \ge
    \sqrt{\frac{m}{2}}.
    $$
    Thus,
    $$
    \mathbf{I}[g]
    \ge \gamma^2 \frac{m\binom{m-1}{k-1}}{2^{m-1}}
    \ge \gamma^2\sqrt{\frac{m}{2}}.
    $$

    Since $|g(x)|\le1$, Parseval's identity gives
    $$
    \sum_{d=0}^m\mathbf{W}^d[g]
    =\mathbb{E}_x[g(x)^2]
    \le1.
    $$
    Fact~\ref{fact:influence-fourier} therefore implies
    $$
    \begin{aligned}
    \mathbf{I}[g]
    &=
    \sum_{d<\ell}d\,\mathbf{W}^d[g]
    +
    \sum_{d\ge\ell}d\,\mathbf{W}^d[g]\\
    &\le
    (\ell-1)\sum_{d<\ell}\mathbf{W}^d[g]
    +m\mathbf{W}^{\ge\ell}[g]\\
    &\le
    (\ell-1)+m\mathbf{W}^{\ge\ell}[g].
    \end{aligned}
    $$
    Finally, $\ell-1\le\gamma^2\sqrt{m/8}$, and hence
    $$
    \begin{aligned}
    \mathbf{W}^{\ge\ell}[g]
    &\ge
    \frac{\gamma^2\sqrt{m/2}-\gamma^2\sqrt{m/8}}{m}\\
    &=
    \frac{\gamma^2}{\sqrt{8m}}.
    \end{aligned}
    $$
\end{proof}

We use the following Fourier-tail reduction of Gretta, Gupta, and Joshi~\cite[Corollary~4.4]{GrettaGuptaJoshi2026FourierConcentrated}.

\begin{fact}\label{fact:fourier-tail-to-fanout}
    Suppose that there is a depth-$d$, $a$-ancilla, single-output $\mathsf{QAC}$ circuit $C$ acting on $n$ inputs. If its associated real-valued function $g_C:\{0,1\}^n \to \mathbb{R}$ satisfies, for some $\eta>0$, $\delta>0$, and integer $\ell \ge n^\delta$,
    $$
    \mathbf{W}^{\ge \ell}[g_C] \ge \eta,
    $$
    then there is a depth-$O(d)$ $\mathsf{QAC}$ circuit $C'$ implementing $\mathtt{FANOUT}_\ell$. The number of ancillae is
    $$
    O\!\left((n+a)\log^2(\ell+1)/\eta\right).
    $$
\end{fact}

\begin{theorem}\label{thm:lower-bound}
    Let $f:\{0,1\}^n\to\{0,1\}$ be a symmetric function with $\rho(f)=k\ge1$. Suppose that a depth-$d$, $a$-ancilla $\mathsf{QAC}$ circuit $C$ computes $f$ with worst-case probability at least $1/2+\varepsilon$, where $\varepsilon\ge1/\log^c n$ for some constant $c>0$. Then $\mathtt{FANOUT}_k$ can be implemented by a depth-$O(d)$ $\mathsf{QAC}$ circuit using $(n+a)^{O(1)}$ ancillae.
\end{theorem}

\begin{proof}
    Set $m \coloneqq 2k-1$. By the definition of $\rho(f)$, either $f_{k-1}\ne f_k$ or $f_{n-k}\ne f_{n-k+1}$. If $f_{n-k}\ne f_{n-k+1}$, negate every input bit of both $f$ and $C$. This preserves the transition radius, circuit parameters, and worst-case probability while moving this transition between weights $k-1$ and $k$, so we may assume without loss of generality that $f_{k-1}\ne f_k$.

    Keep any $m$ input bits and fix the remaining $n-m$ bits to zero. Let $f':\{0,1\}^m\to\{0,1\}$ be the resulting restriction of $f$, and let $C'$ be the corresponding restriction of $C$. Then
    $$
    f'_i = f_i \qquad \text{for all } 0 \le i \le m.
    $$
    In particular, $f'_{k-1} \ne f'_k$, and hence $\rho(f')=k$. Moreover, $C'$ has depth $d$, uses $n+a-m$ ancillae, and computes $f'$ with worst-case probability at least $1/2+\varepsilon$.

    Define $h:\{0,1\}^m\to\{\pm1\}$ by $h(x)\coloneqq(-1)^{f'(x)}$. By the worst-case probability guarantee for $C'$,
    $$
    g_{C'}(x)h(x)\ge2\varepsilon
    \qquad
    \text{for every }x\in\{0,1\}^m.
    $$
    Applying Lemma~\ref{lem:real-valued-central-fourier-tail} to $g_{C'}$ and $h$ with $\gamma\coloneqq2\varepsilon$ gives
    $$
    \mathbf{W}^{\ge\ell}[g_{C'}]
    \ge
    \varepsilon^2\sqrt{\frac{2}{m}},
    $$
    where
    $$
    \ell\coloneqq
    \left\lfloor\varepsilon^2\sqrt{2m}\right\rfloor+1.
    $$

    Fix a constant $A>12c$. If $k\le\log^A n$, the conclusion follows from Fact~\ref{fact:polylog-fanout}. We may therefore assume that $k>\log^A n$. Since $m\ge k$ and $\varepsilon\ge1/\log^c n$, for all sufficiently large $n$ we have
    $$
    \frac{\varepsilon^2\sqrt{2m}}{m^{1/3}}
    =
    \sqrt{2}\,\varepsilon^2m^{1/6}
    \ge
    \sqrt{2}\,\log^{A/6-2c}n
    \ge1.
    $$
    It follows that
    $$
    \ell
    >
    \varepsilon^2\sqrt{2m}
    \ge
    m^{1/3}.
    $$
    Fact~\ref{fact:fourier-tail-to-fanout}, applied to $C'$ with $\delta=1/3$ and
    $$
    \eta
    =
    \varepsilon^2\sqrt{\frac{2}{m}},
    $$
    now gives a depth-$O(d)$ circuit implementing $\mathtt{FANOUT}_\ell$ using
    $$
    O\!\left(
        \frac{(n+a)\sqrt{m}\log^2(\ell+1)}{\varepsilon^2}
    \right)
    $$
    ancillae. Here we used the fact that $C'$ acts on $m$ inputs and has $n+a-m$ ancillae.

    Finally, $\ell^3\ge m\ge k$. A three-level fanout tree, pruned after producing $k$ targets, therefore implements $\mathtt{FANOUT}_k$ using $O(k/\ell)$ applications of $\mathtt{FANOUT}_\ell$. The resulting circuit has depth $O(d)$. Moreover, since $\ell>\varepsilon^2\sqrt{2m}$ and $\ell\le m$, the number of ancillae is at most
    $$
    O\!\left(
        \frac{(n+a)k\log^2(k+1)}{\varepsilon^4}
    \right).
    $$
    Since $1/\varepsilon\le\log^c n$, this quantity is $(n+a)^{O(1)}$, as required.
\end{proof}

\subsection{Consequences}

Theorems~\ref{thm:upper-bound} and~\ref{thm:lower-bound} together prove Theorem~\ref{thm:intro-main}. When $k=0$, the function $f$ is constant and $\mathtt{FANOUT}_0$ is the identity, so both tasks are trivial. For $k\ge1$, Theorem~\ref{thm:upper-bound} computes $f$ exactly using $\mathtt{FANOUT}_k$, while Theorem~\ref{thm:lower-bound} constructs $\mathtt{FANOUT}_k$ from any circuit computing $f$ with worst-case probability at least $1/2+1/\log^c n$, for any fixed constant $c>0$. In particular, $f \in \mathsf{QAC}^0$ if and only if $\mathtt{FANOUT}_k \in \mathsf{QAC}^0$. We now formally restate and prove the two consequences given in the introduction.

\begin{restatedcorollarytwo}
    Let $f:\{0,1\}^n\to\{0,1\}$ be symmetric, and fix any constant $c>0$. If $\rho(f)\ge n^\delta$ for some constant $\delta>0$, then computing $f$ with worst-case probability at least $1/2+1/\log^c n$ is $\mathsf{QAC}^0_{\mathrm{f}}$-complete.
\end{restatedcorollarytwo}

\begin{proof}
    Set $k=\rho(f)$. By Theorem~\ref{thm:lower-bound}, any circuit computing $f$ with worst-case probability at least $1/2+1/\log^c n$ yields a circuit implementing $\mathtt{FANOUT}_k$. Since $k\ge n^\delta$, a fanout tree of depth $O(1/\delta)$ implements $\mathtt{FANOUT}_n$. Conversely, Fact~\ref{fact:symmetric-in-qacf} gives an exact $\mathsf{QAC}^0_{\mathrm{f}}$ computation of $f$, which in particular has the required worst-case probability. Thus, computing $f$ with worst-case probability at least $1/2+1/\log^c n$ and implementing $\mathtt{FANOUT}_n$ are equivalent under $\mathsf{QAC}^0$ reductions.
\end{proof}

\begin{restatedcorollarythree}
    Let $f:\{0,1\}^n\to\{0,1\}$ be symmetric, and suppose that, for some constant $\eta>0$,
    $$
    \widetilde{\deg}(f)=\Omega\!\left(n^{1/2+\eta}\right).
    $$
    Then, for every constant $c>0$, computing $f$ with worst-case probability at least $1/2+1/\log^c n$ is $\mathsf{QAC}^0_{\mathrm{f}}$-complete.
\end{restatedcorollarythree}

\begin{proof}
    The hypothesis implies that $f$ is nonconstant. Fact~\ref{fact:paturi} and the identity $n-\Gamma(f)=2\rho(f)-1$ give
    $$
    \widetilde{\deg}(f)
    =\Theta\!\left(\sqrt{n\bigl(2\rho(f)-1\bigr)}\right).
    $$
    It follows that $\rho(f)=\Omega(n^{2\eta})$. In particular, $\rho(f)\ge n^\eta$ for sufficiently large $n$. The result now follows from Corollary~\ref{cor:intro-qacf-complete}.
\end{proof}

\bibliographystyle{alphaurl}
\bibliography{ref.bib}

\appendix

\section{\texorpdfstring{$\mathtt{EXACT}^n_k$}{EXACT(k)} is in
\texorpdfstring{$\mathsf{QAC}^0[\mathtt{FANOUT}_k]$}{QAC0[FANOUT(k)]}}
\label{app:exact-from-fanout}

In this appendix, we prove Fact~\ref{fact:exact-from-fanout}. The
proof uses the following number-theoretic fact.

\begin{fact}[{\cite[Lemma~1]{HastadWegenerWurmYi1994OptimalDepth}}]\label{fact:number-theoretic}
    Let $\varnothing\ne S\subseteq[n]$. There exists an integer $m$
    with $|S|\le m\le O(|S|^2\log n)$ such that
    $i\not\equiv j\pmod m$ for all distinct $i,j\in S$.
\end{fact}

\begingroup
\renewcommand{\thetheorem}{\ref*{fact:exact-from-fanout}}
\begin{fact}
    Let $k=k(n)$ satisfy $0\le k(n)\le n/2$. Then
    $\mathtt{EXACT}^n_k\in
    \mathsf{QAC}^0[\mathtt{FANOUT}_k]$.
\end{fact}
\endgroup

\begin{proof}
The case $k=0$ is immediate: $\mathtt{EXACT}^n_0$ is the $n$-bit $\mathtt{NOR}$ function, which belongs to $\mathsf{QAC}^0$. We may therefore assume that $k\ge1$. 

We first show that $\mathtt{THRESHOLD}^n_k\in\mathsf{QAC}^0[\mathtt{FANOUT}_k]$. Observe that
$$
\mathtt{FANOUT}_{L} \in \mathsf{QAC}^0[\mathtt{FANOUT}_k] \qquad\text{for every } L=O(k^2\log n).
$$
This follows from the recursive nature of fanout together with Fact~\ref{fact:polylog-fanout}. By Fact~\ref{fact:number-theoretic}, there exists $L=O(k^2\log n)$ such that, for every subset $S\subseteq[n]$ of size $k$, there is some $m\in\{k,\ldots,L\}$ for which the elements of $S$ are pairwise distinct modulo $m$. For each such $m$ and $\ell\in\{0,\ldots,m-1\}$, let
$$
A^m_\ell\coloneqq\{i\in[n]:i\equiv\ell\pmod m\},
\qquad
y^m_\ell\coloneqq\mathtt{OR}(x|_{A^m_\ell}).
$$
Then
$$
\mathtt{THRESHOLD}^n_k(x)
=
\bigvee_{m=k}^{L}
\mathtt{THRESHOLD}^m_k
\bigl(y^m_0,\ldots,y^m_{m-1}\bigr).
$$
Indeed, if $|x|\ge k$, choose any $k$ positions on which $x$ is $1$. By Fact~\ref{fact:number-theoretic}, for some $m\in\{k,\ldots,L\}$ these positions are pairwise distinct modulo $m$, and hence at least $k$ of the corresponding $y^m_\ell$ are $1$. Conversely, if at least $k$ of the $y^m_\ell$ are $1$ for some $m$, then at least $k$ distinct positions of $x$ are $1$.

This identity yields a circuit for $\mathtt{THRESHOLD}^n_k$. Make $L$ copies of each input bit using $\mathtt{FANOUT}_L$ and compute all the $y^m_\ell$ in parallel. By Fact~\ref{fact:symmetric-in-qacf}, all the inner threshold functions can be computed in parallel using fanout gates of size at most $m\le L$. A final $\mathtt{OR}$ gate computes the outer disjunction. Since $L=O(k^2\log n)$ is polynomial in $n$, the entire construction has constant depth and uses polynomially many ancillae. Thus
$$
\mathtt{THRESHOLD}^n_k
\in
\mathsf{QAC}^0[\mathtt{FANOUT}_k].
$$

The same argument, with $k+1$ in place of $k$, gives
$$
\mathtt{THRESHOLD}^n_{k+1}
\in
\mathsf{QAC}^0[\mathtt{FANOUT}_k].
$$
Finally,
$$
\mathtt{EXACT}^n_k
=
\mathtt{THRESHOLD}^n_k
\wedge
\neg\mathtt{THRESHOLD}^n_{k+1},
$$
which proves the fact.
\end{proof}

\end{document}